\documentclass[12pt]{article}

\usepackage{graphics,graphicx,fullpage,natbib,multirow}
\usepackage{amsmath,amssymb,verbatim,epsfig}
\usepackage[dvipsnames,usenames]{color}
\usepackage{caption,subcaption,booktabs}

\newtheorem{proposition}{Proposition}

\newtheorem{defin}{\bf Definition}

\newenvironment{proof}{\noindent{\bf Proof\,}}{$\diamond$}

\def\be{\mbox{Be}}

\def\mul{\mbox{Mul}}

\def\dir{\mbox{Dir}}

\def\E{\mbox{E}}
\def\V{\mbox{Var}}
\def\Cov{\mbox{Cov}}
\def\Cor{\mbox{Corr}}

\def\P{\mbox{P}}

\def\bc{{\bf c}}

\def\bx{{\bf x}}

\def\bF{{\bf F}}

\def\bN{{\bf N}}

\def\bX{{\bf X}}
\def\bY{{\bf Y}}

\def\simind{\stackrel{\mbox{\scriptsize{ind}}}{\sim}}
\def\simiid{\stackrel{\mbox{\scriptsize{iid}}}{\sim}}

\newcommand{\RR}{\mathbb{R}}

\newcommand{\TT}{\mathbb{T}}
\newcommand{\ZZ}{\mathbb{Z}}

\newcommand{\DP}{\mathcal{DP}}
\newcommand{\DDP}{\mathcal{DDP}}

\newcommand{\DMP}{\mathcal{DMP}}
\newcommand{\MP}{\mathcal{MP}}

\newcommand{\BB}{\mathcal{B}}

\newcommand{\XX}{\mathcal{X}}

\begin{document}

\baselineskip=24pt

\title{\bf A class of dependent Dirichlet processes via latent multinomial processes}
\author{{\sc Luis E. Nieto-Barajas} \\[2mm]
{\sl Department of Statistics, ITAM, Mexico} \\[2mm]
{\small {\tt lnieto@itam.mx}} \\}
\date{}
\maketitle

\begin{abstract}
We describe a procedure to introduce general dependence structures on a set of Dirichlet processes. Dependence can be in one direction to define a time series or in two directions to define spatial dependencies. More directions can also be considered. Dependence is induced via a set of latent processes and exploit the conjugacy property between the Dirichlet and the multinomial processes to ensure that the marginal law for each element of the set is a Dirichlet process. Dependence is characterised through the correlation between any two elements. Posterior distributions are obtained when we use the set of Dirichlet processes as prior distributions in a Bayesian nonparametric context. Posterior predictive distributions induce partially exchangeable sequences defined by generalised P\'olya urns. A numerical example to illustrate is also included. 
\end{abstract}

\vspace{0.2in} \noindent {\sl Keywords}: Bayesian nonparametrics, generalised P\'olya urn, moving average process, spatio-temporal process, stationary process.

\section{Introduction}
\label{sec:intro}

The Dirichlet process (DP), introduced by \cite{ferguson:73}, is the most important process prior in Bayesian nonparametric statistics. It is flexible enough to approximate (in the sense of weak convergence) any probability law, although the paths of the process are almost surely discrete \citep{blackwell&macqueen:73}. A DP measure $F$ is characterised by a precision parameter $c>0$ and a centring measure $F_0$ defined on $(\Omega,\BB)$, in notation $F\sim\DP(c,F_0)$. In general, for any $K>0$ and any partition $(B_1,\ldots,B_K)$ of $\Omega$, the random vector $(F(B_1),\ldots,F(B_K))\sim\dir(cF_0(B_1),\ldots,cF_0(B_K))$, that is, has a Dirichlet distribution. 

\cite{sethuraman:94} characterised the DP as an infinite sum of random jumps with probabilities $w_j$ at random locations $\theta_j$, i.e., $F(\cdot)=\sum_{j=1}^\infty w_j\delta_{\theta_j}(\cdot)$, where $\theta_j\simiid F_0$ and $\delta_\theta$ is the dirac measure at $\theta$. This characterisation is also known as stick-breaking since $w_j=v_j\prod_{k<j}(1-v_k)$ with $v_j\simiid\be(1,c)$, that is a beta distribution. 

There have been several proposals to construct dependent measures in Bayesian nonparametrics. They can be classified into three categories: stick breaking, random measures and predictive schemes. We briefly describe them here.

\cite{maceachern:99} introduced a general idea to define dependent Dirichlet processes (DDP) as a way of extending the DP model to multiple random measures $\{F_x\}_{x\in\XX}$. His idea relies on the stick-breaking representation of the DP and introduces dependence by making the probabilities and/or the locations to be a function of $x$, say $w_{x,j}$ and $\theta_{x,j}$, where $x$ is an indexing covariate and $\XX$ is a suitable space. There have been many proposals in the literature to achieve this and most of them have been summarised in \cite{quintana&al:20}. For example, \cite{rodriguez&al:10} rely on latent gaussian copula models to introduce dependence in the locations, and \cite{camerlenghi&al:19a} define the locations as Dirichlet processes or normalised random measures \citep{regazzini&al:03}.

A Dirichlet process can also be seen as a normalisation of a gamma process \citep{ferguson:74}. The more general class of normalised random measures relies on normalising increasing additive processes, where the gamma process is a particular case. A way of introducing dependence, not necessarily in Dirichlet processes, is by making the underlying (unnormalised) random measures to be dependent \citep[e.g.][]{griffin&al:13,lijoi&al:14a}. Clustering performance of a specific class of dependent normalised random measures in mixture models is studied in \cite{lijoi&al:14b}. 

In a different perspective, \cite{walker&muliere:03} defined a DDP via predictive schemes for only two random measures $(F_1,F_2)$. Their approach relies on a latent Dirichlet-multinomial process $N$. Other approaches are those of \cite{mueller&al:04} who introduced dependence by considering convex linear combinations of a common measure and idiosyncratic measures for different studies, and \cite{teh&al:06} who introduced the so called hierarchical Dirichlet process by taking $F_j\mid F_0\sim\DP(c,F_0)$ conditionally independent for a set of $j$'s and $F_0\sim\DP(c_0,G_0)$. Distribution properties of this and more general hierarchical processes is studied in \cite{camerlenghi&al:19b}. 

Posterior predictive distributions of a DP induce exchageable sequences of variables whose law is characterised by a P\'olya urn \citep{blackwell&macqueen:73}. A DDP can also be defined via generalisations of the P\'olya urn to induce partially exchangeable sequences. For instance, \cite{caron&al:07} induce dependence in time via a deletion strategy of past urns, whereas \cite{papas&al:16} and \cite{ascolani&al:21} use a Fleming-Viot process to induce dependence among random measures by sharing a common pool of atoms that depend on a latent death process. 

In this article we construct a class of DDP by considering a predictive scheme, as in \cite{walker&muliere:03}, for multiple random measures that have DP marginal distributions and can be temporal and/or spatial dependent. Predictive distributions of our construction induce partial exchangeable sequences defined by a generalised P\'olya urn. 

The content of the rest of the paper is as follows: In Section \ref{sec:model} we define our generalisation and characterise the dependence induced. In Section \ref{sec:post} we use our model as a Bayesian nonparametric prior distribution,  characterise its posterior laws and define a generalised P\'olya urn to define partially exchangeable sequences. We illustrate the use of the model in Section \ref{sec:illust} and conclude in Section \ref{sec:disc}.

\section{Model construction}
\label{sec:model}

Let us start by defining a multinomial process $N$. This is characterised by an integer parameter $c$ and a measure $F$ defined on $(\Omega,\BB)$, in notation $N\sim\MP(c,F)$, such that for any $K>0$ and for any partition $(B_1,\ldots,B_K)$ of $\Omega$, the random vector $(N(B_1),\ldots,N(B_K))\sim\mul(c;F(B_1),\ldots,F(B_K))$, that is, a multinomial distribution with number of trials $c$ and probabilities $F(B_k)$, $k=1,\ldots,K$. Moreover, if we sample $Y_1,\ldots,Y_c$ independently and identically distributed from $F$, then the multinomial process is defined as $N(\cdot)=\sum_{i=1}^c\delta_{Y_i}(\cdot)$. Moreover, if $F\sim\DP(c_0,F_0)$ and conditionally on $F$, $N\mid F\sim\MP(c,F)$, then marginally $N\sim\DMP(c,c_0,F_0)$, that is, a Dirichlet-multinomial process with parameters $(c,c_0,F_0)$.

We now recall the bivariate Dirichlet process of \cite{walker&muliere:03}. Their construction takes $F_1\sim\DP(c_0,F_0)$, $N\mid F_1\sim\MP(c,F_1)$ and $F_2\mid N\sim\DP(c_0+c,(c_0F_0+N)/(c_0+c))$. Then, after marginalising the latent process $N$, the pair $(F_1,F_2)$ is a DDP with marginal distributions $F_t\sim\DP(c_0,F_0)$ for $t=1,2$ and correlation given by $\Cor\{F_1(B),F_2(B)\}=c/(c_0+c)$ for any set $B\subset\Omega$. 

We note that \cite{walker&muliere:03}'s construction is reversible, so we can start by taking 
\begin{equation}
\label{eq:basic}
N\sim\DMP(c,c_0,F_0)\quad\mbox{and}\quad F_t\mid N\sim\DP(c_0+c,(c_0F_0+N)/(c_0+c)), 
\end{equation}
for $t=1,2$ to obtain the same bivariate model such that $F_t\sim\DP(c_0,F_0)$ marginally. Therefore the key aspect to obtain a Dirichlet process marginal distribution is to condition on a latent $N$ whose law is a Dirichlet-multinomial process. 

We extend this idea to multiple processes $\bF=\{F_t\}_{t\in\TT}$ with $\TT$ a finite index set, say $\TT=\{1,\ldots,T\}$ with $T<\infty$. For that we require a set of latent processes $\bN=\{N_t\}$, one for each $t\in\TT$, plus a single latent process $G$ that will play the role of anchor. Let $\partial_t$ be a set of ``neighbours'', in the broad sense, for each $t\in\TT$. For instance, if $t$ denotes time, we can define dependencies among the $F_t$'s of order $q$ as in time series moving average models. In this case $\partial_t=\{t-q,\ldots,t-1,t\}$. We show in  Figure \ref{fig:temporal} a graphical model with temporal dependence of order $q=2$. Moreover, if $t$ denotes a spatial location, then $\partial_t=\{j:\{j\smile t\}\cup \{j=t\}\}$, where ``$\smile$'' denotes actual spatial neighbour. More general definitions of $\partial_t$ can be taken to define seasonal or spatio-temporal models \citep[see][]{nieto:20}. In any case $t\in\partial_t\subset\TT$ for all $t\in\TT$.

Then, the law of the dependent set $\bF$ is characterised by a three level hierarchical model with the following specifications:
\begin{align}
\nonumber
G&\sim\DP(c_0,F_0)\\
\label{eq:DDP}
N_t\mid G&\simind\MP(c_t,G)\\
\nonumber
F_t\mid\bN&\simind\DP\left(c_0+\sum_{j\in\partial_t} c_j\,,\:\frac{c_0 F_0+\sum_{j\in\partial_t} N_j}{c_0+\sum_{j\in\partial_t} c_j}\right)
\end{align}
for $t\in\TT$, where the set $\bc=\{c_t\}$ is such that $c_0\in\RR^+$, $c_t\in\ZZ^+$ for $t>0$, and take $c_t=0$ and $N_t=0$ with probability one (w.p.1) for $t<0$. In notation we say that $\bF\sim\DDP(\bc,F_0)$. 

Properties of construction \eqref{eq:DDP} are given in Proposition \ref{prop:1}. In particular, the correlation induced and the marginal distributions can be obtained in closed form. 

\begin{proposition}
\label{prop:1}
Let $\{F_t\}_{t\in\TT}$ be a set of dependent measures $\DDP(\bc,F_0)$, defined by  \eqref{eq:DDP}. 
\begin{enumerate} 
\item[(i)] The marginal distribution of $F_t$ is $\DP(c_0,F_0)$ for all $t\in\TT$. 
\item[(ii)] The correlation between any pair of measures $(F_t,F_{t'})$ for $t\neq t'\in\TT$ is given by
$$\Cor\{F_t(B),F_{t'}(B)\}=\frac{c_0\left(\sum_{j\in\partial_{t}\cap\partial_{t'}}c_{j}\right)+\left(\sum_{j\in\partial_{t}}c_{j}\right)\left(\sum_{j\in\partial_{t'}}c_{j}\right)}{\left(c_0+\sum_{j\in\partial_{t}}c_{j}\right)\left(c_0+\sum_{j\in\partial_{t'}}c_{j}\right)},$$
for $B\subset\Omega$, and 
$$\Cor\{F_t(B_i),F_{t'}(B_k)\}=-\sqrt{\frac{F_0(B_i)F_0(B_k)}{\{1-F_0(B_i)\}\{1-F_0(B_k)\}}}\Cor(F_t(B_i),F_{t'}(B_i)),$$
for $B_i,B_k\subset\Omega$ such that $B_i\cap B_k=\emptyset$. 
\item[(iii)] If $c_{t}=0$ for all $t\in\TT$ then the $F_t$'s become independent. 
\end{enumerate}
\end{proposition}
\begin{proof}
\begin{enumerate}
\item[(i)] We note that conditionally on $G$ the $N_t$ are independent multinomial, so the sum is again multinomial, i.e., $\sum_{j\in\partial_{t}}N_{j}\mid G\sim\MP\left(\sum_{(j)\in\partial_{t}}c_{j},G\right)$. Integrating out $G$, the marginal distribution of $\sum_{j\in\partial_{t}}N_{j}$ is Dirichlet-multinomial with parameters $(\sum_{j\in\partial_{t}}c_{j},c_0,F_0)$. Considering level three in \eqref{eq:DDP} and using \eqref{eq:basic} we obtain $F_t\sim\DP(c_0,F_0)$ marginally for $t\in\TT$. 
\item[(ii)] We note that for a single set $B\in\Omega$, $F_t(B)$, $N_t(B)$ and $G(B)$ in \eqref{eq:DDP} simplify to beta, binomial and beta distributions, respectively. We rely on conditional independence properties and the iterative covariance formula. Then for the first part, $\Cov(F_t(B),F_{t'}(B))$ $=\E\{\Cov(F_t(B),F_{t'}(B)\mid\bN)\}+\Cov\{\E(F_t(B)\mid\bN),\E(F_{t'}(B)\mid\bN)\}$. The first term in the sum becomes zero since $F_t(B)$'s are conditionally independent given $\bN$, for $t\neq t'$. The second term, after removing the additive constants of the expected values, is rewritten as $\Cov\left(\sum_{j\in\partial_{t}}N_{j}(B),\sum_{j\in\partial_{t'}}N_{j}(B)\right)$ divided by $(c_0+\sum_{j\in\partial_{t}}c_{j})(c_0+\sum_{j\in\partial_{t'}}c_{j})$. Concentrating in the numerator, and using the iterative covariance formula for a second time, we get $\E\{\Cov(\sum_{j\in\partial_{t}}N_{j}(B),\sum_{j\in\partial_{t'}}N_{j}(B)\mid G)\}+\Cov\{\E(\sum_{j\in\partial_{t}}N_{j}(B)\mid G),\E(\sum_{j\in\partial_{t'}}N_{j}(B)\mid G)\}$. The first term, after separating the sums in the common part, reduces to $\E\{\V(\sum_{j\in\partial_{t}\cap\partial_{t'}}N_{j}(B)\mid G)\}$ which becomes $(\sum_{j\in\partial_{t}\cap\partial_{t'}}c_{j})\E\{G(B)(1-G(B))\}$. After computing the expected value, this can be re-written as $(\sum_{j\in\partial_{t}\cap\partial_{t'}}c_{j})c_0\V(G(B))$. The second term, after computing the expectations becomes $(\sum_{j\in\partial_{t}}c_{j})(\sum_{j\in\partial_{t'}}c_{j})\V(G(B))$. Finally, using $(i)$, we note that $\V(F_{t}(B))=\V(F_{t'}(B))=\V(G(B))$, so dividing the covariance between the product of the standard deviations we obtain the result. For the second part, we proceed analogously, the covariance is the same as before but $\V(G(B))$ is replaced by $\Cov\{G(B_i),G(B_k)\}=-F_0(B_i)F_0(B_k)/(c_0+1)$. Dividing by the product of standard deviations we get the result. 
\item[(iii)] We note that if $c_{t}=0$ then $N_{t}=0$ w.p.1. If we do this for all $t\in\TT$ then it is straightforward to see that the measures $F_{t}$'s become independent. 
\end{enumerate} \vspace{-5mm}\end{proof}

\bigskip
Many things can be concluded from Proposition \ref{prop:1}. Parameter $c_0$ and $F_0$ are the precision and centring measure parameters, respectively, of all Dirichlet process marginal distributions. Parameters $c_t$ for $t>0$ control the strength of the correlation between any two elements $F_t$ and $F_{t'}$, together with the definition of the neighbours $\partial_t$. The correlation is stronger, if the two elements share more parameters. Some locations can be more influential than others if their corresponding $c_{t}$ parameter is larger. Moreover, the set $\{F_{t}\}$ becomes strictly stationary if the parameters $c_t=c$ for $t>0$, which implies that the correlation for $t\neq t'\in\TT$ simplifies to $$\Cor(F_{t}(B),F_{t'}(B))=\frac{r_{\partial_t\cap\partial_{t'}}c_0c+r_{\partial_t}r_{\partial_{t'}}c^2}{(c_0+r_{\partial_t}c)(c_0+ r_{\partial_{t'}}c)},$$
where $r_{\partial_t}$ denotes the number of elements in the set $\partial_t$. 

Let us consider a set with only two elements, $T=2$. There are several ways of defining dependencies between $F_1$ and $F_2$. For instance, an order 1 moving average time series model would have $\partial_1=\{1\}$ and $\partial_2=\{1,2\}$. Then, for a set $B\in\Omega$, $\Cor\{F_1(B),F_2(B)\}=c_1/(c_0+c_1)$, which is the same correlation induced by the bivariate DDP of \cite{walker&muliere:03}. Alternatively, nothing constrains us to define circular dependencies, say $\partial_1=\partial_2=\{1,2\}$. In this case the correlation induced becomes $\Cor\{F_1(B),F_2(B)\}=(c_1+c_2)/(c_0+c_1+c_2)$.

\section{Posterior characterisation and P\'olya urn}
\label{sec:post}

Let us assume that we observe a partially exchangeable sequence in the sense of \cite{definetti:72}. That is, for each $t\in\TT$, we observe a sample of size $m_t$, say $\bX_t=\{X_{1,t},\ldots,X_{m_t,t}\}$ such that $X_{i,t}\mid F_t\simind F_t$ for $i=1,\ldots,m_t$. Moreover, the prior distribution for the set $\bF=\{F_t\}$ is defined by the $\DDP(\bc,F_0)$ given in \eqref{eq:DDP}. 

Since $F_t\mid\bN$'s are independent Dirichlet processes, the conditional posterior distribution for each $F_t$ given $(\bN,G,\bX)$, does not depend on $G$, and can be straightforwardly derived \citep{ferguson:73} 
\begin{equation}
\label{eq:postf}
F_t\mid\bN,\bX_t\sim\DP\left(c_0+\sum_{j\in\partial_t}c_j+m_t\,,\frac{c_0F_0+\sum_{j\in\partial_t}N_j+m_t\widehat{F}_{t}}{c_0+\sum_{j\in\partial_t}c_j+m_t}\right),
\end{equation}
where $\widehat{F}_{t}(\cdot)=\frac{1}{m_t}\sum_{i=1}^{m_t}\delta_{X_{i,t}}(\cdot)$ is the empirical distribution function of sample $\bX_t$, for $t\in\TT$. 

To produce posterior inferences we also need to provide the full conditional distributions for the processes $N_t$ and $G$. For an arbitrary partition $(B_1,\ldots,B_K)$, the posterior conditional distribution of $N_t$ given $(\bF,G,\bX)$, does not depend on $\bX$, and is given by
\begin{eqnarray}
\nonumber
\P\left\{N_t(B_k)=n_{t,k},k=1,\ldots,K\mid\bF,G\right\}\propto\hspace{6cm} \\
\label{eq:postn}
\hspace{2cm}\prod_{k=1}^K\frac{\left\{G(B_k)\prod_{j\in\varrho_t}F_j(B_k)\right\}^{n_{t,k}}}{n_{t,k}!\prod_{j\in\varrho_t}\Gamma\left(c_0F_0(B_k)+\sum_{l\in\partial_j}n_{l,k}\right)}I\left(\sum_{k=1}^K n_{t,k}=c_t\right),
\end{eqnarray}
where $\varrho_t$ is the reversed set of neighbours of $t$, that is, $\varrho_t=\{j:t\in\partial_j\}$. 

Finally, the posterior conditional law for the process $G$ given $(\bN,\bF,\bX)$, does not depend on $\bX$, is another DP of the form
\begin{equation}
\label{eq:postg}
G\mid\bN,\bF\sim\DP\left(c_0+\sum_{t=1}^T c_t\,,\frac{c_0F_0+\sum_{t=1}^T N_t}{c_0+\sum_{t=1}^T c_t}\right).
\end{equation}

With conditional posterior laws $[\bF\mid\bN,\bX]$, $[\bN\mid\bF,G]$ and $[G\mid\bN,\bF]$, given in \eqref{eq:postf}, \eqref{eq:postn} and \eqref{eq:postg}, respectively, we can implement a Gibbs sampler \citep{smith&roberts:93} to obtain posterior summaries. In practice, we choose an arbitrary partition $B_1,\ldots,B_K$ of size $K$ to approximate the paths of the processes. The larger the $K$, the more precise the approximation is. However, for a larger $K$ computational time increases drastically and the mixing of the chain, specially for $N_t(B_k)$'s, is slower since $\sum_{k=1}^K N_t(B_k)=c_t$ w.p.1. 

Posterior predictive distribution of $F_t$ can be easily obtained from \eqref{eq:postf}, conditionally on the latent processes $\bN$. If we re-express the latent processes $N_t$ in terms of random variables $Y_{1,t},\ldots,Y_{c_t,t}$, where $Y_{i,t}\mid G\sim G$, such that $N_t(\cdot)=\sum_{i=1}^{c_t}\delta_{Y_{i,t}}(\cdot)=c_t\widehat{G}_{t}(\cdot)$ then
\begin{equation}
\label{eq:pred}
\E\left(F_t\mid\bN,\bX_t\right)=\frac{c_0F_0+\sum_{j\in\partial_t}c_j\widehat{G}_{j}+m_t\widehat{F}_{t}}{c_0+\sum_{j\in\partial_t}c_j+m_t}.
\end{equation}
Expression \eqref{eq:pred} provides an interesting interpretation of the updated mechanism of our construction. Posterior predictive distribution for $F_t$ is a weighted average of distributions $F_0$ and two sets of empirical distributions: $\widehat{G}_j$ coming from latent data $Y_{i,j}$ for $j\in\partial_t$; and $\widehat{F}_{t}$ coming from actual observed data $X_{i,t}$. The weights are the precision $c_0$, the sample sizes $c_j$ for $j\in\partial_t$, and $m_t$. So, as in the posterior distribution of a DP, parameter $c_0$ is also interpreted as prior sample size. 

Similarly to \cite{blackwell&macqueen:73} and \cite{walker&muliere:03}, we can define the partially exchangeable sequence of observable variables $X_{i,t}$ with P\'olya urns. The dependence mechanism induced by \eqref{eq:DDP} assumes that each urn $t$ starts with a composition of balls $Y_{i,j}$, for $i=1,\ldots,c_j$ and $j\in\partial_t$. Some of these starting balls are shared among urns $t$ and $t'$, according to whether $\partial_t\cap\partial_{t'}$ is empty or not. Then the first draw $X_{1,t}$ is either a new ball coming from $F_0$ with probability $c_0/(c_0+\sum_{j\in\partial_t}c_j)$ or one of the balls in the urn $\bY=\{Y_{i,j}\}$ with probability $\sum_{j\in\partial_t}c_j/(c_0+\sum_{j\in\partial_t}c_j)$. In other words 
$$X_{1,t}\mid\bY\sim\frac{c_0F_0+\sum_{j\in\partial_t}c_j\widehat{G}_{j}}{c_0+\sum_{j\in\partial_t}c_j}.$$
The structure of the urn now changes adding the first ball $X_{1,t}$. The second ball is either a new ball coming from $F_0$, one of the starting balls $\{Y_{i,j}\}$ or the first ball $X_{1,t}$, that is
$$X_{2,t}\mid X_{1,t},\bY\sim\frac{c_0F_0+\sum_{j\in\partial_t}c_j\widehat{G}_{j}+\delta_{X_{1,t}}}{c_0+\sum_{j\in\partial_t}c_j+1}.$$
We add the already extracted balls to the urn. In general, for the $m_t+1$ ball we have 
$$X_{m_t+1,t}\mid X_{1,t},\ldots,X_{m_t,t},\bY\sim\frac{c_0F_0+\sum_{j\in\partial_t}c_j\widehat{G}_{j}+m_t\widehat{F}_{t}}{c_0+\sum_{j\in\partial_t}c_j+m_t}.$$

\section{Illustration}
\label{sec:illust}

The Official Statistics Institute of Mexico constructs an economic activity index (ITAEE) for each of the 32 states of the country and it is reported every three months. The last report of this index (https://www.inegi.org.mx/temas/itaee/) consists in destationalized  values, where seasonal peaks have been removed, of constant prices of the year 2013. For the purpose of our analysis we consider the indexes of the last six years available, that is, from the first trimester of year 2015 to the third trimester of year 2020. In total we have samples $\{X_{i,t}\}$ for $i=1,\ldots,m_t$ of sizes $m_t=32$, for $t=1,\ldots,T$ with $T=23$ trimesters. 

The data is reported in Figure \ref{fig:series}. We can clearly see that in the second trimester of 2020 there is a drop for all states due to the Covid-19 pandemic. Previous to year 2020, the state of Baja California Sur (top green) showed its highest value in the third trimester of 2018. On the other hand, states like Campeche (bottom blue) and Tabasco (bottom green) showed a decreasing trend in its economy. 

The objective of our analysis is to characterise the variability of the economy in the whole country. For that we use our $\DDP$ model with order $q$ temporal dependencies such that the neighbouring sets are defined as $\partial_t=\{t-q,\ldots,t-1,t\}$. We took $c_0=0.1$ and $F_0(x)=\Phi((x-\mu_0)/\sigma_0)$ with $\Phi$ the standard normal cumulative distribution function. If we denote by $X_{(1)}$ and $X_{(m)}$ the sample minimum and maximum, respectively, then we took $\mu_0=(X_{(1)}+X_{(m)})/2$ and $\sigma_0=(X_{(m)}-X_{(1)})/7$. We defined a partition of size $K=50$ with $B_k=(b_{k-1},b_{k}]$ and $b_0=X_{(1)}$ and $b_k=b_{k-1}+(X_{(m)}-X_{(1)})/K$. The dependence parameters $c_t$ were assumed constant across time. For $c_t$ and for the order of dependence $q$ we consider a range of different values to compare. In particular we took $c_t\in\{5,10,15,20\}$ and $q\in\{1,2,\ldots,9\}$. 

To choose among the different model specifications, we considered two statistics. The logarithm of the pseudo marginal likelihood \citep{geisser&eddy:79} defined as $$\mbox{LPML}=\frac{1}{T}\sum_{t=1}^T\frac{1}{m_t}\sum_{i=1}^{m_t}\mbox{CPO}_{i,t}\quad\mbox{with}\quad\mbox{CPO}_i=\left\{\frac{1}{L}\sum_{l=1}^L\frac{1}{f_t^{(l)}(x_{i,t})}\right\}^{-1},$$ where $\mbox{CPO}_{i,t}$ is the Monte Carlo approximation of the conditional predictive ordinate and $l$ denotes the iteration. The second statistic is the L-measure \citep{ibrahim&laud:94} which is a summary between variance and bias of the predictive distribution. This is defined as $$\mbox{LMEA}(\nu)=\frac{1}{TK}\sum_{t=1}^T\sum_{k=1}^{K}\V\{F_t(B_k)\mid\bx\}+\frac{\nu}{TK}\sum_{t=1}^T\sum_{k=1}^K\left[E\{F_t(B_k)\mid\bx\}-\widehat{F}_t(B_k)\right]^2$$
with $\nu\in[0,1]$. 

We implemented a Gibbs sampler with the full conditional distributions described in Section \ref{sec:post}. Sampling from \eqref{eq:postf} and \eqref{eq:postg} is straightforward since the finite dimensional distributions for $(F_t(B_1),\ldots,F_t(B_K))$ and $(G(B_1),\ldots,G(B_K))$, conditionally on the other processes, become Dirichlet distributions. However, sampling from \eqref{eq:postn} requires the implementation of a Metropolis-Hastings step \citep{tierney:94} and sampling each of the $N_t(B_k)$, one at a time, for $k=1,\ldots,K-1$ and $t=1,\ldots,T$. This procedure induces highly autocorrelated chains. To overcome this effect we ran large chains with thinning. In particular we took 100,000 iterations with a burn-in of 5,000 and a thinning of one of every 25$^{th}$ iteration.

The goodness of fit statistics, LPML and LMEA for $\nu=1/2$, for the different combinations of $q$ and $c_t$ are reported in Figure \ref{fig:gof}. We recall that higher LPML and smaller LMEA denote a better model. Considering LPML statistic, we note that for $c_t>5$, the best model is obtained with $q=1$, however, for $c_t=15$, models with $q$ between 5 and 8 also obtain a high LPML. Considering now LMEA statistic, for a fixed $c_t$ the best model is obtained for $q>1$. In particular, for $c_t=15$ the best model is that with $q=6$, so combining both goodness of fit measures, we take this latter as our best model. Since the data are observations every three months, having a temporal dependence of order 6 means that the distribution of the economy in a given time depends on the economy of up to around one year and a half ago.

For the best fitting model we characterise posterior predictive distributions and report the posterior mean as point estimates for each of the $F_t$'s. These are included in Figure \ref{fig:cdf}. In this graph we use darker colours to indicate larger (more recent) times. For early times (small $t$'s) the cumulative distribution functions (CDF) are more concentrated around the $x$ values of 110, whereas for more recent times (large $t$'s) the CDFs show more variability. In particular, there are two times that show a distinctive behaviour and assign larger probabilities to smaller $x$ values. These two correspond to the second and third trimester of 2020 where the Mexican economy was shaken by the Covid-19 pandemic. 

Finally, we also show in Figure \ref{fig:cdf} in red colour, the posterior mean of the anchoring process $G$. It can be interpreted as the overall mean behaviour of the Mexican economy in the period of study.

\section{Discussion}
\label{sec:disc}

We have introduced a collection of dependent Dirichlet processes via a hierarchical model. Three levels were needed to ensure that the marginal distribution for each element is a Dirichlet process. Temporal, spatial or temporal-spatial dependencies are possible. Other types of dependencies, like circular dependencies mentioned at the end of Section \ref{sec:model}, are also possible. 

Posterior inference of the collection of processes is possible when used as Bayesian nonparametric prior distributions. This requires an easy to implement Gibbs sampler. 

One of the common uses of Dirichlet processes is to define mixtures for model based clustering, exploiting the discreteness of the Dirichlet process. We can also use our $\DDP$ for that purpose. Say $X_{i,t}\mid\theta_{i,t}\sim h(x_{i,t}\mid\theta_{i,t})$, with $h$ a probability density, then $\theta_{i,t}\mid F_t\simind F_t$ and $\bF\sim\DDP(\bc,F_0)$ for $i=1,\ldots,m_t$ and $t=1,\ldots,T$. Clusters obtained for each time $t$ will be dependent, according to the chosen definition of the neighbouring sets $\partial_t$. Studying the performance of this mixture model is worthy, but it is left for future work.

\section*{Acknowledgements}

The author acknowledges support from \textit{Asociaci\'on Mexicana de Cultura, A.C.}

\bibliographystyle{natbib}

\newpage

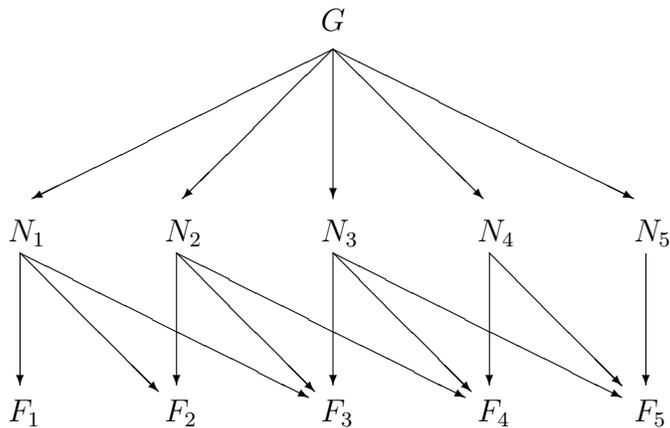
\begin{figure}
\setlength{\unitlength}{0.8cm}
\begin{center}
\begin{picture}(20,8)
\put(9.8,7.6){$G$} 
\put(10,7.3){\vector(-2,-1){5}}
\put(10,7.3){\vector(-1,-1){2.5}}
\put(10,7.3){\vector(0,-1){2.5}}
\put(10,7.3){\vector(1,-1){2.5}}
\put(10,7.3){\vector(2,-1){5}}
\put(4.6,4.1){$N_1$}
\put(7.2,4.1){$N_2$}
\put(9.8,4.1){$N_3$}
\put(12.4,4.1){$N_4$}
\put(15,4.1){$N_5$}
\put(4.6,1.1){$F_1$}
\put(7.2,1.1){$F_2$}
\put(9.8,1.1){$F_3$}
\put(12.4,1.1){$F_4$}
\put(15,1.1){$F_5$}
\put(4.8,3.9){\vector(0,-1){2.2}}
\put(4.8,3.9){\vector(1,-1){2.3}}
\put(4.8,3.9){\vector(2,-1){4.8}}
\put(7.4,3.9){\vector(0,-1){2.2}}
\put(7.4,3.9){\vector(1,-1){2.3}}
\put(7.4,3.9){\vector(2,-1){4.8}}
\put(10,3.9){\vector(0,-1){2.2}}
\put(10,3.9){\vector(1,-1){2.3}}
\put(10,3.9){\vector(2,-1){4.8}}
\put(12.6,3.9){\vector(0,-1){2.2}}
\put(12.6,3.9){\vector(1,-1){2.2}}
\put(15.2,3.9){\vector(0,-1){2.2}}
\end{picture}
\end{center}
\vspace{-1cm}
\caption{Graphical representation of temporal dependence of order $q=2$.}\label{fig:temporal} 
\end{figure}

\begin{figure}
\centerline{\includegraphics[scale=0.8]{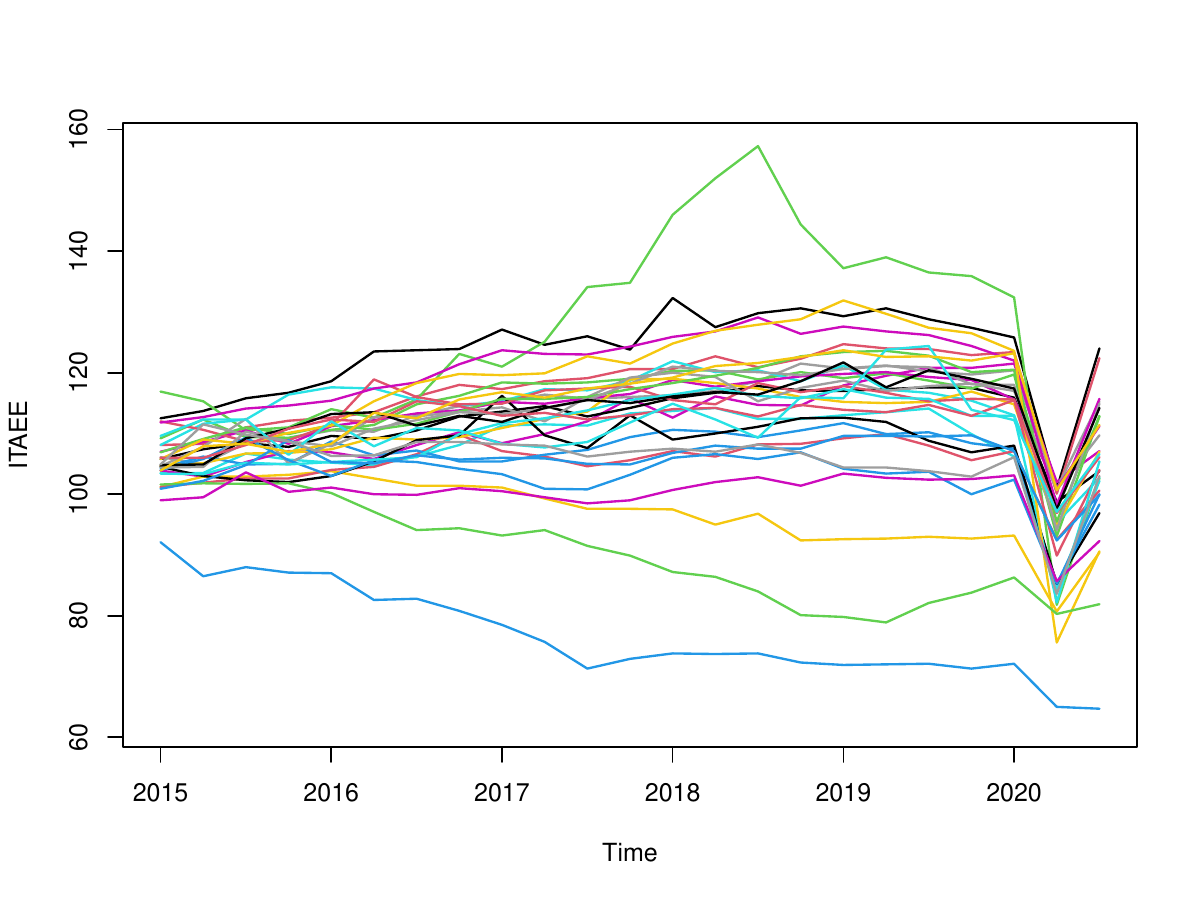}}
\vspace{-2mm}
\caption{{\small Time series of ITAEE index for the 32 states of Mexico from the first trimester of 2015 to the third trimester of 2020.}}
\label{fig:series}
\end{figure}

\begin{figure}
\centerline{\includegraphics[scale=0.69]{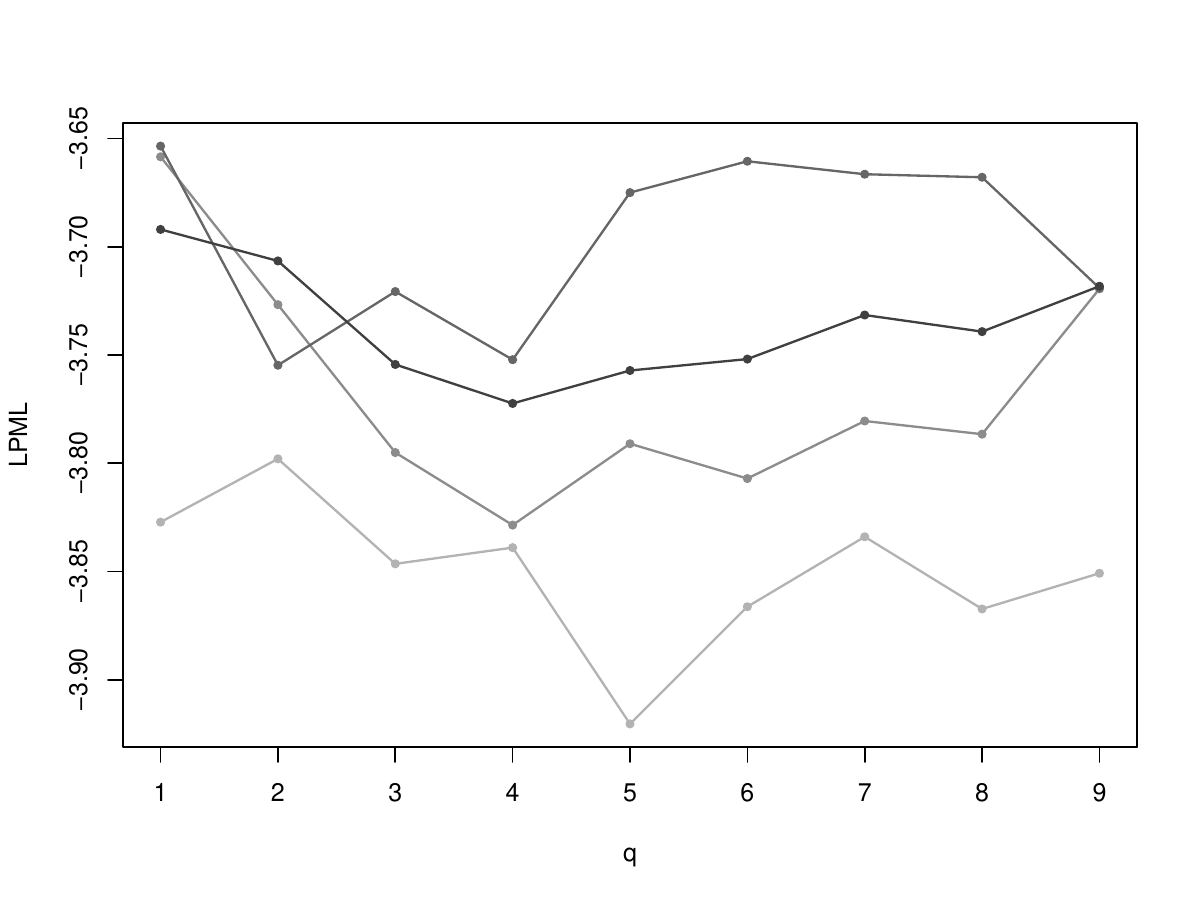}}
\vspace{-7mm}
\centerline{\includegraphics[scale=0.69]{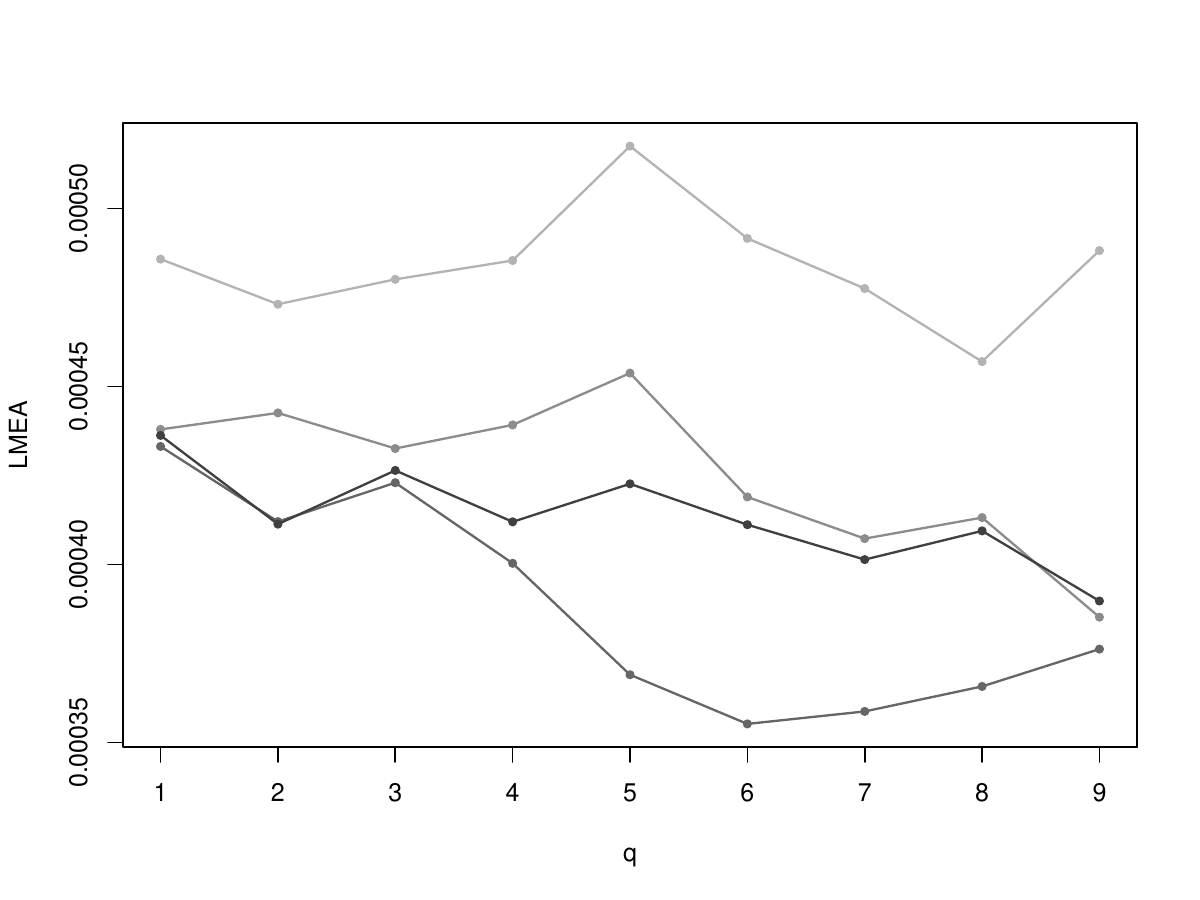}}
\vspace{-2mm}
\caption{{\small LPML and LMEA statstics for $q=1,\ldots,9$ and $c_t=5,10,15,20$. Darker lines correspond to larger values of $c_t$.}}
\label{fig:gof}
\end{figure}

\begin{figure}
\centerline{\includegraphics[scale=0.8]{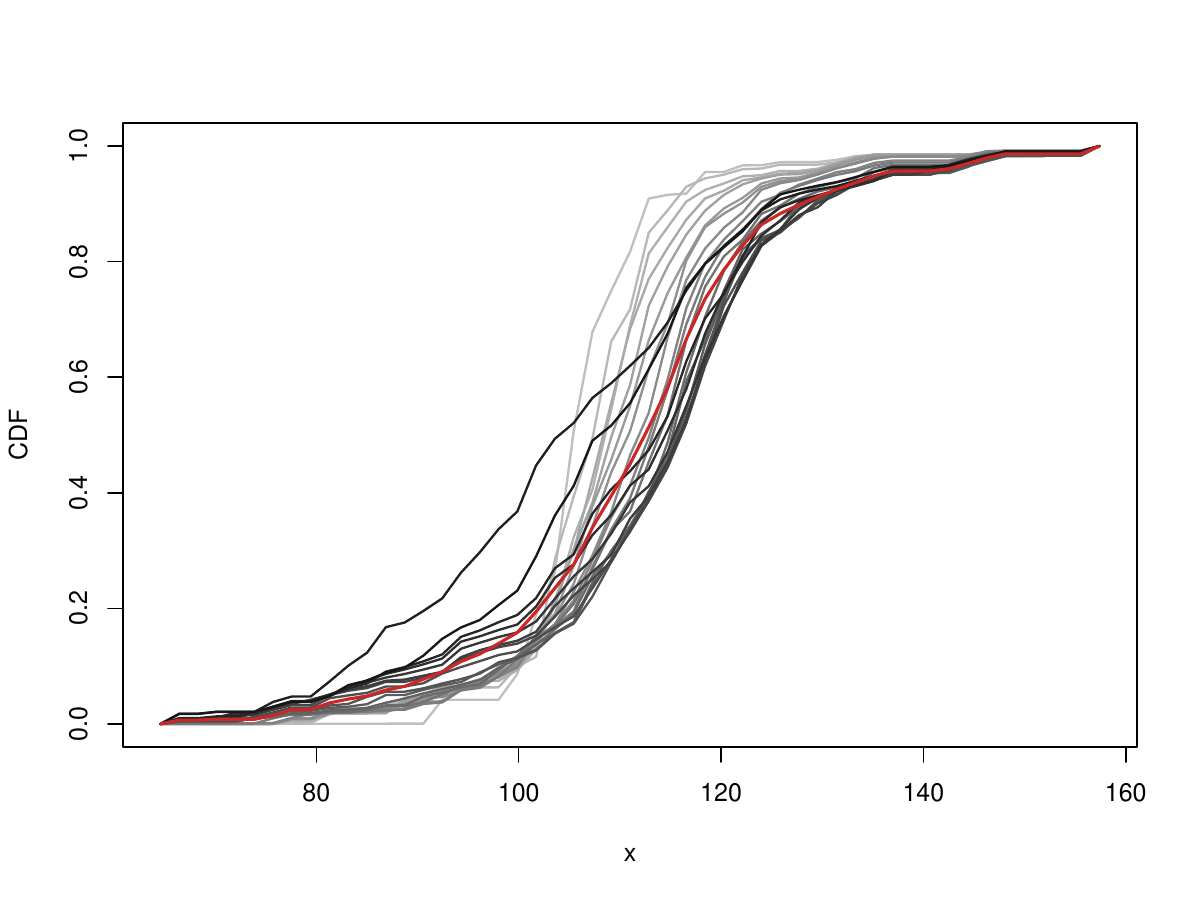}}
\vspace{-2mm}
\caption{{\small Posterior estimates of $F_t$ for $t=1,\ldots,23$. Darker lines correspond to larger times. Posterior estimate of $G$ shown in red.}}
\label{fig:cdf}
\end{figure}

\end{document}